%% file: main.tex
\documentclass[11pt]{article}

\usepackage[margin=1in]{geometry}
\usepackage{amsmath,amssymb,amsthm}
\usepackage{booktabs,tabularx,array}
\usepackage{microtype}
\usepackage{tikz}
\usetikzlibrary{arrows.meta,backgrounds}
\usepackage[colorlinks=true,linkcolor=blue,citecolor=blue,urlcolor=blue]{hyperref}
\usepackage{caption}
\usepackage{enumitem}
\usepackage{float}
\usepackage{placeins}

\input{macros}
\hypersetup{
  pdftitle={A Certified Lower Bound for Lebesgue's Universal Cover Problem},
  pdfauthor={Niantao Xie},
  pdfsubject={Finite certificate for a convex universal-cover lower bound},
  pdfkeywords={Lebesgue universal cover problem, convex geometry, finite certificate, interval arithmetic}
}

\renewcommand{\arraystretch}{1.14}

\newcolumntype{Y}{>{\raggedright\arraybackslash}X}
\newcolumntype{L}[1]{>{\raggedright\arraybackslash}p{#1}}
\newcolumntype{Q}[1]{>{\raggedleft\arraybackslash}p{#1}}
\setlist[itemize]{topsep=3pt,itemsep=2pt,parsep=0pt,leftmargin=2.2em}
\setlist[enumerate,1]{label=(\roman*),leftmargin=2.8em,labelwidth=2.1em,labelsep=0.6em,align=left,itemindent=0pt,topsep=3pt,itemsep=2pt,parsep=0pt}
\allowdisplaybreaks[2]
\newtheorem{proposition}{Proposition}
\newtheorem{theorem}{Theorem}
\newtheorem{corollary}{Corollary}
\newtheorem{lemma}{Lemma}
\newtheorem{definition}{Definition}

\title{A Certified Lower Bound for Lebesgue's Universal Cover Problem}
\author{Niantao Xie}
\date{}

\begin{document}
\maketitle

\begin{abstract}
Lebesgue's universal cover problem asks for the infimum of the areas of planar sets that contain a congruent copy of every planar set of diameter at most one.  We work in the convex Brass--Sharifi three-test-set framework, where the test sets are a closed disk, an equilateral triangle, and a regular pentagon of diameter one.  For each normalized placement \(\paramv\), let \(\Area(\paramv)\) denote the area of the convex hull of these three test sets.  We construct a finite certificate proving
\[
\Area(\paramv)\ge 0.833
\]
throughout the admissible normalized domain.  The threshold \(0.833\) yields a certified improvement over the Brass--Sharifi lower bound \(0.832\) within the same convex three-test-set framework.

The proof is a finite-cover argument.  The admissible domain is covered by finitely many parameter domains, and each domain carries a local lower-bound certificate.  The non-witness domains are certified by supporting local records.  On the witness domains, the local bound is obtained from an inner-witness polygon construction.  The witness points lie in the three test sets and determine an ordered polygonal region certified to be simple and positively oriented.  Its area is bounded below by interval orientation and shoelace estimates.  Since this certified polygonal region lies inside the corresponding convex hull, its area gives a lower bound for the hull area.  Combining the local inequalities with the finite cover yields
\[
\acvx\ge 0.833,
\]
where \(\acvx\) is the infimum of the areas of convex universal covers.
\end{abstract}

\section{Introduction}

A planar set is a universal cover if it contains a congruent copy of every planar set of diameter at most one.  In this paper the covering set is required to be convex.  Let
\[
\Ucvx=
\left\{K\subset\R^2:
\begin{array}{l}
K\text{ is convex, and for every }S\subset\R^2\text{ with }\diam(S)\le 1,\\
\text{there is a Euclidean isometry }g\text{ such that }g(S)\subseteq K
\end{array}
\right\}.
\]
Set
\[
\acvx=\inf_{K\in\Ucvx}\area(K).
\]
Here and below, \(\area\) denotes planar Lebesgue area.  All lower-bound statements below concern this convex quantity.  For background on the general universal cover problem, see Baez, Bagdasaryan, and Gibbs~\cite{baez-bagdasaryan-gibbs}; a small-area convex covering construction is described by Gibbs~\cite{gibbs}.

Brass and Sharifi introduced the three-test-set framework used here and obtained the lower bound \(0.832\) in this convex setting~\cite{brass-sharifi}.  For comparison on the upper-bound side, Gibbs constructed a convex universal cover of area \(0.8440935944\)~\cite{gibbs}.  The test sets are a closed disk, an equilateral triangle, and a regular pentagon of diameter one.  Their normalization turns the relative placement problem into a finite-dimensional problem for the area of the convex hull of these three sets.  The present work certifies the threshold
\[
\target=0.833
\]
within the same normalized convex setting.  We regard \(0.833\) as an exact decimal threshold.

The absolute increment from \(0.832\) to \(0.833\) is \(0.001\).  Relative to the interval between the Brass--Sharifi lower bound \(0.832\) and the convex upper bound constructed by Gibbs, \(0.8440935944\), it closes approximately \(8.3\%\) of that gap.  The finite certificate gives a reproducible finite-record verification of the threshold.  Its records include supporting local certificates for the non-witness domains, explicit witness-domain polygon certificates, interval orientation and shoelace lower endpoints, and a final finite-cover aggregation.  The verification code replays these finite checks deterministically.

The structure of the proof is as follows.  First, a convex universal cover contains congruent copies of the three test sets, and convexity forces it to contain their convex hull.  Second, the Brass--Sharifi normalization represents the hull by a parameter \(\paramv\) in an admissible domain \(\Oadm\).  Third, the certificate covers \(\Oadm\) by finitely many domains and verifies a local lower bound on each one.  Fourth, the finite-cover implication gives \(\Area(\paramv)\ge\target\) throughout \(\Oadm\), and taking the infimum over convex universal covers gives \(\acvx\ge\target\).

\emph{Outline.} Section~2 defines the normalized three-test-set problem.  Section~3 states the finite certificate model.  Section~4 explains the sources of local lower bounds.  Section~5 gives the inner-witness polygon construction.  Section~6 describes the interval orientation and shoelace estimates.  Section~7 proves the certified lower-bound theorem and the convex consequence.  Section~8 concludes and records the scope of the statement.  Appendices~A and~B summarize the finite certificate data used by the proof.

\section{The normalized three-test-set problem}

\subsection{The three test sets}

Let \(\Cdisc\), \(\Ttri\), and \(\Pfive\) be, respectively, fixed reference copies of a closed disk, an equilateral triangle, and a regular pentagon of diameter one.  The disk has radius \(1/2\).  The reference disk is centered at the origin.  The reference triangle and pentagon have fixed reference positions and orientations.  The three sets \(\Cdisc,\Ttri,\Pfive\) are the necessary test sets used in the lower-bound argument: every universal cover contains a congruent copy of each of them.  Thus, if \(K\in\Ucvx\), then after choosing one congruent copy of each test set inside \(K\), convexity of \(K\) implies that \(K\) contains the convex hull of those three chosen copies.

\subsection{Normalized placements}

Following the Brass--Sharifi normalization, the disk is fixed at the origin and the orientation of the triangle is fixed.  The regular pentagon is taken in a fixed reference orientation; its rotational symmetry permits the pentagon angle to be recorded in a fundamental interval.  Let \(R_\rho\) denote counterclockwise rotation by angle \(\rho\) about the origin.  Let
\[
\paramv=(\rho,x_3,y_3,x_5,y_5),\qquad
\transTri=(x_3,y_3),\qquad
\transPent=(x_5,y_5).
\]
Here \(\transTri\) is the triangle translation, \(\rho\) is the pentagon rotation, and \(\transPent\) is the pentagon translation.

Set
\[
X(\paramv)=\Cdisc\cup(\Ttri+\transTri)\cup(R_\rho\Pfive+\transPent),
\qquad
\Hull(\paramv)=\conv\bigl(X(\paramv)\bigr),
\]
and set
\[
\Area(\paramv)=\area(\Hull(\paramv)).
\]
Here \(\conv\) denotes the closed convex hull.  Thus \(\Area\) denotes the hull-area function, while \(\area\) denotes planar area.  Figure~\ref{fig:placement} illustrates the normalized placement.

\input{figures/fig_placement}
\FloatBarrier

\subsection{The admissible domain and normalization principle}

Denote by \(\Oadm\) the admissible normalized placement domain in the Brass--Sharifi framework.  This is a recorded subset of the five-dimensional parameter space with coordinates \((\rho,x_3,y_3,x_5,y_5)\).  Its points represent normalized relative placements satisfying the Brass--Sharifi admissibility conditions after the normalization and symmetry reductions used for the three test sets.

\paragraph{Input normalization principle.}
The present paper uses the Brass--Sharifi normalization principle~\cite{brass-sharifi} as an input theorem and treats \(\Oadm\) as the input domain for the finite certificate.  It does not rederive the full symbolic reduction from the unreduced placement space.  For every \(K\in\Ucvx\), after applying a global Euclidean isometry if necessary, there exists \(\paramv\in\Oadm\) such that
\[
\Hull(\paramv)\subseteq K.
\]
Consequently, a uniform lower bound for \(\Area(\paramv)\) on \(\Oadm\) gives a lower bound for \(\area(K)\) for every convex universal cover \(K\).  This is the only point at which the Brass--Sharifi normalization enters the final universal-cover consequence.  Thus the finite certificate begins after the normalization and symmetry reductions of Brass--Sharifi have been accepted as input.  The finite cover of \(\Oadm\) is introduced in Section~\ref{sec:certificate-model}, and the certificate proves
\[
\Area(\paramv)\ge\target\qquad(\paramv\in\Oadm).
\]
The convex universal-cover consequence is proved in Section~\ref{sec:theorem}.

\section{The finite certificate model}\label{sec:certificate-model}

\begin{definition}[Finite certificate]
A finite certificate for the threshold \(\target\) consists of the following finite data.
\begin{enumerate}
\item A finite family \(\Bcover\) of parameter domains in the normalized placement space.
\item A cover assertion \(\Oadm\subseteq\bigcup_{B\in\Bcover}B\).
\item For each \(B\in\Bcover\), a local value \(L_B\) satisfying \(\Area(\paramv)\ge L_B\ge\target\) on \(B\).
\item Individual evidence records certifying the local assertions.
\item Outward-rounded interval enclosures for the determinant and area estimates used by the certificate.
\item A final aggregation check connecting the local assertions to the convex lower-bound statement.
\end{enumerate}
\end{definition}

\subsection{Certificate cover}

Let \(\Bcover\) be the finite cover supplied by the certificate.  It is a finite family of parameter domains in the five-dimensional normalized placement space, and each domain \(B\in\Bcover\) is represented by interval bounds on the placement parameters.  Its mathematical role is the inclusion
\[
\Oadm\subseteq \bigcup_{B\in\Bcover}B.
\]
The symbol \(\Bcover\) always denotes the cover family, while \(B\) denotes one of its domains.

\subsection{Local lower-bound assertions}

For each \(B\in\Bcover\), the local assertion has the form
\[
\Area(\paramv)\ge L_B\qquad(\paramv\in B),
\]
where the certified lower endpoint also satisfies
\[
L_B\ge\target.
\]
Section~\ref{sec:local-sources} separates the cover according to how the value \(L_B\) is certified; both subfamilies enter the final proof through the same inequality.

\subsection{Evidence records and interval enclosures}

Each local assertion is accompanied by finite evidence.  The record families are summarized in Table~\ref{tab:record-families}.  The names ``directed interval'', ``tensor'', ``bridge'', and ``witness'' are record-class labels in the finite archive; their role in the proof is only the local inequality displayed in the last column.

\begin{table}[H]
\centering
\caption{Certificate record families and their local lower-bound interfaces.}
\label{tab:record-families}
\small
\begin{tabularx}{0.98\textwidth}{
    @{\hspace{0.02\textwidth}}          
    m{0.27\textwidth}                   
    @{\hspace{0.03\textwidth}}          
    X                                   
    @{\hspace{0.03\textwidth}}          
    >{\raggedleft\arraybackslash}m{0.20\textwidth}   
    @{\hspace{0.02\textwidth}}          
}
\toprule
Record family & Role in the proof & Local interface \\
\midrule
directed interval records & Direct interval lower bounds on supporting domains. & \(\Area(\paramv)\ge L_B\ge\target\) \\
tensor records & Auxiliary finite local-bound records on supporting domains. & \(\Area(\paramv)\ge L_B\ge\target\) \\
bridge records & Supporting records for residual domains. & \(\Area(\paramv)\ge L_B\ge\target\) \\
witness records & Ordered witness polygons and interval area lower bounds. & \(\Area(\paramv)\ge L_B^{\mathrm{wit}}\ge\target\) \\
\bottomrule
\end{tabularx}
\end{table}

Here \(L_B^{\mathrm{wit}}\) denotes the witness-domain lower endpoint introduced in Section~\ref{sec:local-sources}.  All supporting record families enter the proof only through certified values \(L_B\ge\target\).  The use of individual evidence records ties each local assertion in the cover to a checked record, rather than only to an aggregate count.

The numerical records use outward-rounded interval arithmetic.  Each coordinate, determinant, and shoelace expression is evaluated in an interval enclosure rounded outward, so a recorded lower endpoint is a rigorous lower bound under the stated interval model.  This is the standard principle of interval arithmetic; see Moore~\cite{moore}.

The sizes of the finite cover and the supporting evidence records are summarized in Table~\ref{tab:component-summary}.  The finite certificate conditions used in the final aggregation are listed in Table~\ref{tab:obligations}.  These tables summarize checked certificate data; the proof uses the corresponding local inequalities, not the table entries as independent assumptions.  The counts in Appendix~\ref{app:data} occur at different levels: parameter domains are cover-level objects, whereas final verification records are records used to discharge the certificate obligations.  They are not intended to be in one-to-one correspondence.  The aggregation step is isolated as Proposition~\ref{prop:finite-cover}.  We next separate the local assertions according to the mechanism that supplies \(L_B\).

\section{Sources of local lower bounds}\label{sec:local-sources}

Let \(\Fwit\subseteq\Bcover\) be the subfamily of witness domains, and set
\[
\Fsup=\Bcover\setminus\Fwit.
\]
For \(B\in\Fsup\), the value \(L_B\) is supplied by supporting interval records and the associated individual evidence records.  For \(B\in\Fwit\), the value \(L_B\) is realized by a witness lower bound \(L_B^{\mathrm{wit}}\).  The distinction between \(\Fsup\) and \(\Fwit\) concerns only the method used to certify the local lower bound; both families enter the final proof through the same inequality \(\Area(\paramv)\ge L_B\ge\target\).

\subsection{Supporting domains}

On the supporting domains \(B\in\Fsup\), each supporting record certifies an interval lower bound for \(\Area(\paramv)\) on its domain and checks that the certified lower endpoint is at least the same threshold \(\target\) used in the theorem.  The final proof uses the local inequalities attached to these records, not their aggregate counts.  A supporting local record enters the finite cover only after the inequalities \(\Area(\paramv)\ge L_B\) on \(B\) and \(L_B\ge\target\) have been checked.

\subsection{Witness domains}

On the witness domains \(B\in\Fwit\), the local value is obtained from an inner-witness polygon.  In this case the local lower bound is the witness lower bound,
\[
L_B=L_B^{\mathrm{wit}},
\qquad
\Area(\paramv)\ge L_B^{\mathrm{wit}}
\qquad(\paramv\in B),
\]
and the interval estimates of Section~\ref{sec:interval} establish \(L_B^{\mathrm{wit}}\ge\target\).  The witness-domain counts and interval lower-bound minima are summarized in Tables~\ref{tab:witness-counts} and~\ref{tab:witness-bounds}.

\section{The inner-witness polygon construction}\label{sec:witness}

We describe the geometric construction used on the witness domains.

\subsection{Witness point families}

Fix \(B\in\Fwit\).  A witness family on \(B\) is an ordered list of point functions
\[
q_1(\paramv),q_2(\paramv),\ldots,q_{m_B}(\paramv).
\]
Here \(m_B\) denotes the number of witness vertices stored for the domain \(B\).  We denote the underlying set of these witness points by
\[
Q_B(\paramv)=\{q_1(\paramv),\ldots,q_{m_B}(\paramv)\}.
\]
The containment condition is
\[
q_i(\paramv)\in X(\paramv)
\qquad(1\le i\le m_B,\ \paramv\in B).
\]
Each witness point has recorded provenance: it is a translated vertex of \(\Ttri\), a rotated-translated vertex of \(\Pfive\), or a certified point of \(\Cdisc\), meaning a point whose membership in the diameter-one disk is verified by outward-rounded interval enclosures.  The triangle witness coordinates are affine in the translation variables.  The pentagon witness coordinates are trigonometric functions of the rotation angle \(\rho\) and affine functions of the translation \(\transPent\).  In all cases, the coordinate enclosures on \(B\) are evaluated by outward-rounded interval arithmetic.

\subsection{The ordered witness polygon}

The certificate also stores a prescribed cyclic order of the witness points.  Let \(\mathcal P_B(\paramv)\) denote the compact polygonal region bounded by the ordered vertices
\[
q_1(\paramv),q_2(\paramv),\ldots,q_{m_B}(\paramv).
\]
The certificate includes a convex-order audit for this prescribed order.  In this context, the audit is a finite check that fixes the cyclic boundary order of the recorded vertices and rules out order changes or crossing degeneracies on the whole parameter domain.  In the archive this audit is stored as finite sign and separation data attached to the same parameter domain \(B\).  The determinant estimates of Section~\ref{sec:interval} are not used alone to infer global simplicity from consecutive turns; rather, they certify the stability of this audited order.  Together, the convex-order audit and the determinant estimates certify that the ordered polygonal region is simple and positively oriented throughout \(B\).

All vertices of \(\mathcal P_B(\paramv)\), equivalently all points of \(Q_B(\paramv)\), lie in \(X(\paramv)\subseteq\Hull(\paramv)\).  Since \(\Hull(\paramv)\) is convex and \(\mathcal P_B(\paramv)\) is the polygonal region bounded by these ordered vertices,
\[
\mathcal P_B(\paramv)\subseteq\conv Q_B(\paramv)\subseteq\Hull(\paramv).
\]

\subsection{The inner-witness lower-bound lemma}

\begin{lemma}[Inner-witness lower bound]\label{lem:inner-witness}
Suppose that the ordered witness polygonal region \(\mathcal P_B(\paramv)\) is contained in \(\Hull(\paramv)\) for every \(\paramv\in B\), and that
\[
\area(\mathcal P_B(\paramv))\ge L_B^{\mathrm{wit}}
\qquad(\paramv\in B).
\]
Then
\[
\Area(\paramv)\ge L_B^{\mathrm{wit}}
\qquad(\paramv\in B).
\]
\end{lemma}

\begin{proof}
For \(\paramv\in B\), the containment \(\mathcal P_B(\paramv)\subseteq\Hull(\paramv)\) and area monotonicity give
\[
\Area(\paramv)=\area(\Hull(\paramv))
\ge \area(\mathcal P_B(\paramv))
\ge L_B^{\mathrm{wit}}.
\]
\end{proof}

\section{Interval orientation and shoelace estimates}\label{sec:interval}

\subsection{Coordinate enclosures}

Write
\[
q_i(\paramv)=(x_i(\paramv),y_i(\paramv)),
\qquad 1\le i\le m_B.
\]
The certificate evaluates interval enclosures for these coordinate functions on each witness domain \(B\).  The enclosures are rounded outward at every operation.  Consequently, all determinant and shoelace lower endpoints used below are conservative.

\subsection{Cyclic order and determinant lower bounds}

For consecutive triples in the prescribed cyclic order stored in the certificate data, define
\[
\Delta_{B,i}(\paramv)
=
\detv\bigl(q_{i+1}(\paramv)-q_i(\paramv),\,
q_{i+2}(\paramv)-q_{i+1}(\paramv)\bigr),
\]
with indices taken cyclically.  The interval evaluation supplies lower endpoints \(\underline\Delta_{B,i}\) such that
\[
\underline\Delta_{B,i}\le \Delta_{B,i}(\paramv)
\qquad(\paramv\in B).
\]
The convex-order audit fixes the cyclic order.  For the estimates required by that prescribed order, the recorded condition is
\[
\underline\Delta_{B,i}>0.
\]
These determinant conditions certify that the prescribed order remains positively oriented on \(B\).  They are used together with the convex-order audit, not as an independent global-simplicity theorem.  No shoelace lower bound is accepted unless the convex-order audit and the corresponding determinant lower bounds have both been certified.

\subsection{Shoelace lower endpoints}

For the certified order, set
\[
S_B(\paramv)=
\sum_{i=1}^{m_B}
\bigl(x_i(\paramv)y_{i+1}(\paramv)-x_{i+1}(\paramv)y_i(\paramv)\bigr),
\qquad q_{m_B+1}=q_1
\]
by cyclic convention; equivalently, \((x_{m_B+1},y_{m_B+1})=(x_1,y_1)\).  Since \(\mathcal P_B(\paramv)\) is certified to be a simple positively oriented polygonal region, the shoelace formula gives
\[
\area(\mathcal P_B(\paramv))=\frac12 S_B(\paramv).
\]
The interval computation supplies a lower endpoint \(\underline S_B\) satisfying
\[
\underline S_B\le S_B(\paramv)
\qquad(\paramv\in B).
\]
Thus
\[
\area(\mathcal P_B(\paramv))\ge \frac12\underline S_B.
\]
If \(\frac12\underline S_B\ge\target\), Lemma~\ref{lem:inner-witness} gives \(\Area(\paramv)\ge\target\) on \(B\).  Thus every witness domain carries a verified local inequality \(\Area(\paramv)\ge\target\), placing it in the local framework of Section~3.2.

\section{The certified lower-bound theorem}\label{sec:theorem}

\subsection{Finite-cover proposition}

\begin{proposition}[Finite-cover implication]\label{prop:finite-cover}
Let \(\Bcover\) be a finite family of parameter domains such that
\[
\Oadm\subseteq\bigcup_{B\in\Bcover}B.
\]
Assume that for every \(B\in\Bcover\) there is a certified number \(L_B\) satisfying
\[
\Area(\paramv)\ge L_B\ge\target
\qquad(\paramv\in B).
\]
Then \(\Area(\paramv)\ge\target\) for every \(\paramv\in\Oadm\).
\end{proposition}

\begin{proof}
Let \(\paramv\in\Oadm\).  By the finite-cover relation, \(\paramv\in B\) for at least one \(B\in\Bcover\).  The local certificate on that domain gives
\[
\Area(\paramv)\ge L_B\ge\target.
\]
Since \(\paramv\) was arbitrary, the inequality holds on \(\Oadm\).
\end{proof}

\subsection{Verification of the bundled certificate}

The labels OB-A--OB-F refer to the finite proof obligations listed in Table~\ref{tab:obligations}.

\begin{proposition}[Verification of the bundled certificate]\label{prop:bundled-certificate}
The certificate archive accompanying this paper satisfies the finite certificate conditions OB-A--OB-F listed in Table~\ref{tab:obligations} for the threshold \(\target\).
\end{proposition}

\begin{proof}
The accompanying finite records are organized into four components summarized in Tables~\ref{tab:component-summary}--\ref{tab:witness-bounds}:
\begin{itemize}
\item per-record evidence, which ties supporting local records to individual evidence;
\item construction records, which provide the supporting construction and rounding checks;
\item witness records, which provide point containment, convex-order audit, determinant lower bounds, and shoelace lower endpoints on the witness domains;
\item final verification records, which connect the cover-level proof obligations to the stated convex claim boundary.
\end{itemize}
The public verification commands check the fixed SHA256 manifest and the following recorded predicates: required archive members; schema and row-count invariants; per-record pass fields; positivity of recorded surplus fields where they occur; zero final blockers; and agreement of the final aggregation gate with the stated convex claim boundary.  These are finite-record checks on the archived certificate data; they do not regenerate the geometric search that produced the records.  These finite predicates are the certificate-side representatives of OB-A--OB-F, and the bundled records satisfy them for \(\target\).
\end{proof}

\subsection{Certified lower-bound theorem}

\begin{theorem}[Certified lower-bound theorem]\label{thm:certificate}
For the bundled finite certificate, one has
\[
\Area(\paramv)\ge\target
\qquad(\paramv\in\Oadm).
\]
\end{theorem}

\begin{proof}
By Proposition~\ref{prop:bundled-certificate}, the bundled certificate satisfies the cover condition and the local lower-bound conditions.  For \(B\in\Fsup\), the supporting records give \(\Area(\paramv)\ge L_B\ge\target\) on \(B\).  For \(B\in\Fwit\), Lemma~\ref{lem:inner-witness} and the interval orientation and shoelace conditions give \(\Area(\paramv)\ge L_B^{\mathrm{wit}}\ge\target\) on \(B\).  Thus every domain in \(\Bcover\) satisfies the local hypothesis of Proposition~\ref{prop:finite-cover}, and the theorem follows.
\end{proof}

\subsection{Convex universal-cover consequence}

\begin{corollary}[Convex universal-cover consequence]\label{cor:convex}
The bundled finite certificate gives
\[
\acvx\ge\target=0.833.
\]
\end{corollary}

\begin{proof}
Let \(K\in\Ucvx\).  By the normalization principle of Section~2.3, choose \(\paramv\in\Oadm\) such that \(\Hull(\paramv)\subseteq K\).  Then
\[
\area(K)\ge\area(\Hull(\paramv))=\Area(\paramv).
\]
Theorem~\ref{thm:certificate} gives \(\Area(\paramv)\ge\target\).  Hence \(\area(K)\ge\target\).  Taking the infimum over \(K\in\Ucvx\) yields the desired inequality.
\end{proof}

\section{Conclusion}

The finite certificate proves \(\Area(\paramv)\ge\target\) on the admissible normalized three-test-set domain, with \(\target=0.833\).  The result is obtained in the same convex three-test-set framework as the Brass--Sharifi threshold \(0.832\).

The proof proceeds by finite covering.  Each parameter domain carries either a supporting local record or an ordered inner-witness polygonal record with interval orientation and shoelace lower bounds.  The finite-cover proposition then yields the convex universal-cover consequence \(\acvx\ge0.833\).

The statement proved here is the convex lower-bound consequence obtained within the stated Brass--Sharifi three-test-set certificate model.  It is not a statement about the unrestricted nonconvex problem, and it is not a proof-assistant formalization.

\paragraph{Certificate archive and verification code.}
The certificate records and verification code accompanying this paper are available from the project repository.\footnote{\url{https://github.com/Sheldon-Anderson/lebesgue-universal-cover-new-lower-bound-certificate}}

\section*{Acknowledgements}
The author used software-assisted editorial tools for copyediting.  All mathematical claims, computations, certificate judgments, and final text are the responsibility of the author.

\clearpage
\appendix
\numberwithin{table}{section}
\renewcommand{\arraystretch}{1.15}

\section{Finite certificate data}
\label{app:data}

\begin{table}[H]
\centering
\caption{Finite certificate components.}
\label{tab:component-summary}
\small
\begin{tabularx}{0.94\textwidth}{
    @{\hspace{0.02\textwidth}}          
    m{0.23\textwidth}                   
    @{\hspace{0.12\textwidth}}          
    >{\centering\arraybackslash}m{0.13\textwidth}   
    @{\hspace{0.02\textwidth}}          
    >{\raggedleft\arraybackslash}X      
    @{\hspace{0.02\textwidth}}          
}
\toprule
Component & Quantity & Role \\
\midrule
parameter domains & 356,816 & cover of \(\Oadm\) \\
directed interval records & 41,261 & supporting local bounds \\
tensor records & 8,751 & auxiliary local bounds \\
bridge records & 282 & residual supporting local bounds \\
witness domains & 16 & witness-domain local bounds \\
\bottomrule
\end{tabularx}
\end{table}

The labels OB-A--OB-F denote the proof obligations discharged by the final aggregation record.

\begin{table}[H]
\centering
\caption{Finite certificate conditions.}
\label{tab:obligations}
\small
\begin{tabularx}{0.97\textwidth}{
    @{\hspace{0.02\textwidth}}
    m{0.10\textwidth}                           
    @{\hspace{0.0\textwidth}}
    X                                           
    @{\hspace{0.05\textwidth}}
    X                                           
    @{\hspace{0.02\textwidth}}
}
\toprule
Group & Mathematical assertion & Certificate component / verification record \\
\midrule
OB-A & \(\Oadm\subseteq\bigcup_{B\in\Bcover}B\) & finite-cover records \\
OB-B & supporting domains satisfy local lower bounds & per-record evidence and construction audit \\
OB-C & witness domains satisfy local lower bounds & witness construction \\
OB-D & every used local record has individual evidence & per-record evidence \\
OB-E & interval endpoints clear \(\target\) & construction, witness, and final verification interval records \\
OB-F & final aggregation matches the convex claim & final verification, final aggregation, and claim-boundary records \\
\bottomrule
\end{tabularx}
\end{table}

\begin{table}[H]
\centering
\caption{Final verification record summary.}
\label{tab:final-adjudication-summary}
\small
\begin{tabularx}{0.97\textwidth}{
    @{\hspace{0.02\textwidth}}          
    m{0.26\textwidth}                   
    @{\hspace{0.06\textwidth}}          
    >{\centering\arraybackslash}m{0.12\textwidth}   
    @{\hspace{0.02\textwidth}}          
    >{\raggedleft\arraybackslash}X      
    @{\hspace{0.02\textwidth}}          
}
\toprule
Record class & Quantity & Role \\
\midrule
frozen certificate records & 3,411 & final record family used by the final verification \\
final-replay records & 2,790 & records requiring only final replay \\
event-aware interval records & 493 & records certified by event-aware interval replay \\
witness/source records & 117 & witness and source-reconstruction records \\
thin extra records & 11 & near-critical records retained separately \\
final blockers & 0 & unresolved records after final verification \\
\bottomrule
\end{tabularx}
\end{table}
\FloatBarrier

\clearpage
\section{Witness-domain interval summaries}
\label{app:inner-witness}

\begin{table}[H]
\centering
\caption{Witness-domain counts.}
\label{tab:witness-counts}
\small
\begin{tabularx}{0.72\textwidth}{
    @{\hspace{0.02\textwidth}}          
    >{\raggedright\arraybackslash}X     
    @{\hspace{0.03\textwidth}}          
    >{\raggedleft\arraybackslash}m{0.22\textwidth}   
    @{\hspace{0.02\textwidth}}          
}
\toprule
Quantity & Value \\
\midrule
witness domains & 16 \\
accepted terminal subdomains & 140 \\
unresolved terminal subdomains & 0 \\
witness point incidences & 2,112 \\
point-containment certificates & 2,112 / 2,112 \\
\bottomrule
\end{tabularx}
\end{table}

\begin{table}[H]
\centering
\caption{Interval lower-bound summaries on the witness domains.}
\label{tab:witness-bounds}
\small
\begin{tabularx}{0.82\textwidth}{
    @{\hspace{0.02\textwidth}}          
    >{\raggedright\arraybackslash}X     
    @{\hspace{0.03\textwidth}}          
    >{\raggedleft\arraybackslash}m{0.24\textwidth}   
    @{\hspace{0.02\textwidth}}          
}
\toprule
Quantity & Value \\
\midrule
minimum witness-domain area bound & \(0.8642876791\) \\
minimum excess over \(\target\) & \(0.0312876791\) \\
minimum orientation determinant lower endpoint & \(3.6637\times10^{-5}\) \\
\bottomrule
\end{tabularx}
\end{table}

\noindent The minima in Table~\ref{tab:witness-bounds} are taken only over the witness domains.  They are not asserted to be the global infimum of \(\Area(\paramv)\) over \(\Oadm\); the theorem uses only that every accepted local record clears the threshold \(\target\).
\FloatBarrier

\end{document}

%% file: macros.tex
\newcommand{\R}{\mathbb{R}}
\newcommand{\Ucvx}{\mathcal{U}_{\mathrm{cvx}}}
\newcommand{\acvx}{\alpha_{\mathrm{cvx}}}
\newcommand{\Cdisc}{C}
\newcommand{\Ttri}{T}
\newcommand{\Pfive}{P_5}
\newcommand{\paramv}{v}
\newcommand{\transTri}{u_3}
\newcommand{\transPent}{u_5}
\newcommand{\Hull}{H}
\newcommand{\Area}{A}
\newcommand{\Oadm}{\Omega_{\mathrm{adm}}}
\newcommand{\Bcover}{\mathcal{F}}
\newcommand{\Fwit}{\mathcal{F}_{\mathrm{wit}}}
\newcommand{\Fsup}{\mathcal{F}_{\mathrm{sup}}}
\newcommand{\target}{\tau}
\newcommand{\conv}{\operatorname{conv}}
\newcommand{\area}{\operatorname{area}}
\newcommand{\diam}{\operatorname{diam}}
\newcommand{\detv}{\det}

%% file: figures/fig_placement.tex
\begin{figure}[H]
\centering
\begin{tikzpicture}[scale=4.25, line cap=round, line join=round, >=Latex, every node/.style={font=\small}]
  \coordinate (O) at (0,0);
  \coordinate (Tctr) at (0.14,0.22);
  \coordinate (Pctr) at (0.52,-0.18);

  \coordinate (A) at (-0.22,0.10);
  \coordinate (B) at (0.40,0.10);
  \coordinate (Ctop) at (0.09,0.64);

  \foreach \i in {0,...,4} {
    \coordinate (P\i) at ({0.52+0.34*cos(72*\i+30)}, {-0.18+0.34*sin(72*\i+30)});
  }

  \def\HullTopAngle{139.2063}
  \def\HullBotAngle{252.1517}
  \coordinate (HtopDisk) at ({0.35*cos(\HullTopAngle)}, {0.35*sin(\HullTopAngle)});
  \coordinate (HbotDisk) at ({0.35*cos(\HullBotAngle)}, {0.35*sin(\HullBotAngle)});

  \def\PlacementHullPath{%
    (HtopDisk)
    -- (Ctop)
    -- (P0)
    -- (P4)
    -- (P3)
    -- (HbotDisk)
    arc[start angle=\HullBotAngle,end angle=\HullTopAngle,radius=0.35]%
  }

  \draw[->, gray!55] (-0.78,-0.70) -- (1.08,-0.70) node[right, font=\scriptsize] {$x$};
  \draw[->, gray!55] (-0.78,-0.70) -- (-0.78,0.90) node[above, font=\scriptsize] {$y$};

  \begin{scope}[on background layer]
    \fill[gray!5] \PlacementHullPath -- cycle;
  \end{scope}

  \fill[gray!6] (O) circle (0.35);
  \draw[thick] (O) circle (0.35);
  \draw[thin] (-0.35,0)--(0.35,0);
  \fill (O) circle (0.012);
  \node[anchor=north east] at (O) {$O$};
  \node[anchor=east] at (-0.12,-0.19) {$C$};

  \fill[gray!24] (A)--(B)--(Ctop)--cycle;
  \draw[thick] (A)--(B)--(Ctop)--cycle;
  \fill (Tctr) circle (0.011);
  \draw[->, thick] (O) -- (Tctr);
  \node[font=\scriptsize, anchor=south west] at (0.11,0.21) {$u_3$};
  \node[anchor=west, align=left, font=\scriptsize] at (-0.05,0.40) {$T+u_3$};

  \fill[gray!38] (P0)--(P1)--(P2)--(P3)--(P4)--cycle;
  \draw[thick] (P0)--(P1)--(P2)--(P3)--(P4)--cycle;
  \fill (Pctr) circle (0.011);
  \draw[->, thick] (O) -- (Pctr);
  \node[font=\scriptsize, anchor=north west] at (0.40,-0.18) {$u_5$};
  \node[anchor=west, align=left, font=\scriptsize] at (0.35,-0.37) {$R_\rho P_5+u_5$};

  \draw[->, thick] (0.52,-0.18) ++(0.19,0) arc[start angle=0,end angle=30,radius=0.19];
  \node[font=\scriptsize] at (0.75,-0.06) {$\rho$};

  \draw[very thick, dashed, gray!80] \PlacementHullPath -- cycle;
  \node[align=center, font=\scriptsize] at (-0.3,0.45) {$H(v)$};

\end{tikzpicture}

\caption{Schematic normalized placement of the three Brass--Sharifi test sets. The disk $\Cdisc$ is fixed at the origin $O$, the triangle is translated by $\transTri=(x_3,y_3)$, and the pentagon is rotated counterclockwise about the origin by $\rho$ and then translated by $\transPent=(x_5,y_5)$. The enclosing curve represents $\Hull(\paramv)=\conv(\Cdisc\cup(\Ttri+\transTri)\cup(R_\rho\Pfive+\transPent))$.}
\label{fig:placement}
\end{figure}
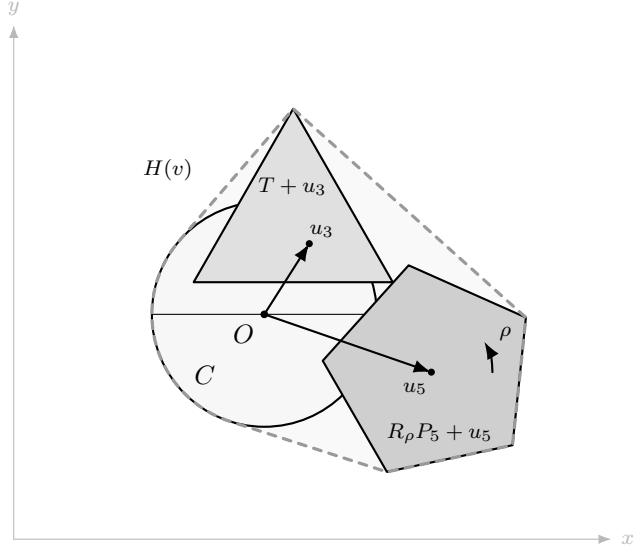